\documentclass[journal,draftcls, onecolumn]{IEEEtran}
\usepackage[dvipdfmx]{graphicx}
\usepackage{amsmath}
\usepackage{amssymb}
\usepackage{amsfonts}
\usepackage{mathrsfs}
\usepackage{theorem}
\usepackage{color}
\DeclareMathAlphabet{\bm}{OML}{cmm}{b}{it}
\usepackage[ruled,vlined]{algorithm2e}
\usepackage{algorithmic}

\theorembodyfont{\rmfamily}
\newtheorem{theorem}{Theorem}
\newtheorem{lemma}{Lemma}
\newtheorem{definition}{Definition}
\newtheorem{corollary}{Corollary}
\newtheorem{remark}{Remark}
\newtheorem{proposition}{Proposition}

\newcommand{\qed}{\hfill \IEEEQED}

\newcommand{\bol}[1]{\mathbf{#1}}

\begin{document}

\title{
Isomorphism Problem Revisited: 
Information Spectrum Approach
%Invariance of Information Spectrum
%Information Spectral Analysis on Interval Algorithm for Random Number Generation
%under Conditional Independence Structure
}

\author{Shun Watanabe and Te Sun Han}

% The paper headers
%\markboth{Journal of \LaTeX\ Class Files,~Vol.~6, No.~1, January~2007}%
%{Shell \MakeLowercase{\textit{et al.}}: Bare Demo of IEEEtran.cls for Journals}

\maketitle
\begin{abstract}
The isomorphism problem in the ergodic theory is revisited from the perspective of
information spectrum approach, an approach that has been developed 
to investigate coding problems for non-ergodic random processes in information theory.
It is proved that the information spectrum is invariant under isomorphisms.
This result together with an analysis of information spectrum provide a conceptually simple
proof of the result by \v{S}ujan, which claims that the entropy spectrum is invariant under isomorphisms.
%which can be regarded as a generalization of the classic result, by Kolmogorov and Sinai, 
%claiming that the entropy is invariant under isomorphisms. 
It is also discussed under what circumstances the same information spectrum implies the existence 
of an isomorphism. 
\end{abstract}

%%%%%%%%%%%%%%%%%%%%%%%%%%%%%%%%%%%%%%%%%%%%%%%%%%%%%%%%%%%%%%%%%%%%%%%%%%
\section{Introduction} \label{section:introduction}

In ergodic theory, one of fundamental problems is to identify if two dynamical systems are isomorphic
or not, which is known as the isomorphism problem. Inspired by the Shannon entropy in information theory,
Kolmogorov and Sinai introduced the entropy of dynamical systems and showed that the entropy is an invariant
under isomorphism; in other words, an isomorphism between two dynamical systems exists only if the entropies are equal.
Since then, the entropy has been widely used as an invariant of the isomorphism problem. 
However, the entropy need not be a complete invariant, i.e., an isomorphism may not exist even if the entropies of two 
dynamical systems are equal.
Then, an interesting question is under what circumstances the same entropy implies the existence of an isomorphism. 
A landmark result on this problem was provided by Ornstein in \cite{Ornstein:70} (see also \cite{Ornstein:book}).
He proved that two i.i.d. random processes (Bernoulli shifts) are isomorphic to each other if the entropies are equal;
furthermore, he also characterized the class of processes that are isomorphic to the i.i.d. random processes.
See \cite{Shields-book, Shields:98} for other interactions between information theory and ergodic theory.

In the literature, most studies on the isomorphism problem have focused on ergodic dynamical systems with
some exceptions \cite{KieRah81, Sujan:82, Sujan:82b, TakVer:02}.
In \cite{KieRah81}, Kieffer and Rahe provided a sufficient condition on the existence of isomorphism 
between two non-ergodic mixtures of Bernoulli shifts.
In \cite{Sujan:82, Sujan:82b},  \v{S}ujan provided a necessary condition for the existence of
isomorphism in terms of ``entropy spectra," leveraging the ergodic decomposition. 
In \cite{TakVer:02}, Takens and Verbitskiy showed that the R\'enyi entropy of non-ergodic
dynamical system is given by the essential infimum of the spectrum of entropies of the ergodic decomposition. 

On the other hand, in the 1990s, Han and Verd\'u developed ``information spectrum" approach
in information theory to investigate coding problems for
general non-ergodic sources/channels \cite{han:93} (see also \cite{han:book}). 
Among other things, the key feature of the approach is that coding theorems are proved in two steps.
In the first step, the performance of a coding problem is characterized by 
the probabilistic behavior of self-information or related quantities,
which is termed the information spectrum. This step is proved without invoking
probability theoretic theorems, such as the law of large number or the ergodic theorem. 
Then, in the second step, the probability theoretic theorems are invoked to characterize the
behavior of the information spectrum in terms of information measures such as the entropy. 

The main aim of this paper is to revisit the isomorphism problem 
from the perspective of information spectrum approach. 
More specifically, we prove that the information spectrum is invariant under isomorphisms between 
random processes.\footnote{It should not be confused with the spectral isomorphism of linear operators induced by dynamical systems (eg.~see \cite{Walters:book}).} 
Then, using this result together with an analysis of information spectra,
we provide an alternative proof for the aforementioned result by \v{S}ujan \cite{Sujan:82, Sujan:82b},
which is conceptually and technically much simpler than the argument given in \cite{Sujan:82, Sujan:82b}.\footnote{An advantage of the approach in \cite{Sujan:82, Sujan:82b} is that
it can be applied to random processes with countably infinite alphabet, while we use the finiteness of alphabet in our proof.}

Even though  the information spectrum coincides with the entropy
spectrum under ergodic decomposition, we are intentionally distinguishing the two concepts, ``information spectrum" and ``entropy spectrum."
The former is defined directly for a given random process, and we prove the invariance of information 
spectra without invoking the ergodic decomposition; the ergodic decomposition is only needed to prove that
the information spectrum coincides with the entropy spectrum. On the other hand, the argument in \cite{Sujan:82, Sujan:82b} begins with the ergodic decomposition,
and the invariance of the entropy spectrum is proved via the invariance of entropy in each
ergodic component. 

The rest of the paper is organized as follows. In Section \ref{sec:preliminaries}, we introduce our notation and
review some basic facts in ergodic theory. In Section \ref{section:information-spectrum}, we state our main
results; the proofs are provided in Section \ref{section:proof-explicit-information-spectrum} and Section \ref{section:proof-spectrum-homomorphism}.
In Section \ref{sec:general-dynamical-system}, we discuss how to define the information spectrum for general dynamical systems. 
In Section \ref{section:sufficient-condition}, we discuss under what circumstances the same information spectrum
implies the existence of an isomorphism. 
In Section \ref{section:discussion}, we conclude the paper with some discussion on possible future research directions. 

%%%%%%%%%%%%%%%%%%%%%%%%%%%%%%%%%%%%%%%%%%%%%%%%%%%%%%%%%%%%%%%%%%%%%%%%%%%
%%%%%%%%%%%%%%%%%%%%%%%%%%%%%%%%%%%%%%%%%%%%%%%%%%%%%%%%%%%%%%%%%%%%%%%%%%%
\section{Preliminaries} \label{sec:preliminaries}

In this section, we introduce our notation by reviewing some basic facts in ergodic theory.
Let $(\Omega,\mathscr{B},\mu)$ be a measure space.
A measurable map $T: \Omega \to \Omega$ is called measure-preserving transformation if
$\mu(T^{-1}A) = \mu(A)$ for every $A \in \mathscr{B}$.
The quadruple $(\Omega, \mathscr{B},\mu,T)$ is called a dynamical system.
When $\Omega = {\cal X}^{\mathbb{Z}}$, i.e., 
the set of all doubly infinite sequences
\begin{align*}
\bm{x} = (\ldots,x_{-2},x_{-1},x_0,x_1,x_2,\ldots),
\end{align*}
where each $x_i$ is an element of some finite set ${\cal X}$,
the measure-preserving transformation is given by the shift $S$,\footnote{Since the underlying space can be recognized from
the context, we denote $S$ instead of $S_{\cal X}$ to avoid cumbersome notation.} 
i.e., $(S \bm{x})_i = x_{i+1}$ for $\bm{x} \in {\cal X}^{\mathbb{Z}}$; the measurable set $\mathscr{B}_{\cal X}$
is given by the $\sigma$-algebra generated by cylinder set
\begin{align}
[a_m^n] := \{ \bm{x} \in {\cal X}^{\mathbb{Z}} : x_i = a_i~\forall m \le i \le n \}
\end{align}
for $m,n\in \mathbb{Z}$.
Let us define the random process $\bm{X}=\{X_n \}_{n \in \mathbb{Z}}$ by assigning 
\begin{align}
P_{X_m^n}(a_m^n) = \Pr\big( X_i = a_i : m \le i \le n \big) = \mu([a_m^n])
\end{align}
for $m,n \in \mathbb{Z}$.
Owing to the measure-preserving requirement of $S$, the random process $\bm{X}$ is stationary. 
When $m=1$, we denote $P_{X_m^n}(a_m^n)$ by $P_{X^n}(a^n)$ for $a^n \in {\cal X}^n$.
In this manner, the dynamical system $({\cal X}^{\mathbb{Z}}, \mathscr{B}_{\cal X}, \mu,S)$ can be identified with the random process $\bm{X}$.
Throughout the rest of this paper except Section \ref{sec:general-dynamical-system}, 
we mainly consider the random process $\bm{X}$ determined by $({\cal X}^{\mathbb{Z}}, \mathscr{B}_{\cal X}, \mu,S)$;
we will come back to general dynamical systems in Section \ref{sec:general-dynamical-system}.

One of the most fundamental problems in ergodic theory is to determine if given two processes are ``equivalent" or not.
A commonly used notion of equivalence is defined as follows. 

\begin{definition}
For two stationary random processes $\bm{X}=\{X_n\}_{n\in\mathbb{Z}}$ and $\bm{Y}=\{Y_n\}_{n\in\mathbb{Z}}$ 
determined by $({\cal X}^{\mathbb{Z}}, \mathscr{B}_{\cal X}, \mu,S)$ and $({\cal Y}^{\mathbb{Z}}, \mathscr{B}_{\cal Y},\nu,S)$,
respectively, 
we call a measurable map $\phi: {\cal X}^{\mathbb{Z}} \to {\cal Y}^{\mathbb{Z}}$ a {\em homomorphism}\footnote{Homomorphism is also called {\em factor map} in some literature.}
if $\nu = \phi_* \mu$, i.e., $\nu(B)=\mu(\phi^{-1}(B))$ for every $B \in \mathscr{B}_{\cal Y}$,
and $\phi(Sx)=S\phi(x)$ for almost sure $\bm{x} \in {\cal X}^{\mathbb{Z}}$ under $\mu$.
Furthermore, when there exists a homomorphism $\psi:{\cal Y}^{\mathbb{Z}} \to {\cal X}^{\mathbb{Z}}$ such that
$\psi(\phi(\bm{x}))=\bm{x}$ for almost sure $\bm{x} \in {\cal X}^{\mathbb{Z}}$ under $\mu$ and $\phi(\psi(\bm{y}))=\bm{y}$ for almost sure 
$\bm{y} \in {\cal Y}^\mathbb{Z}$ under $\nu$, then a pair $(\phi,\psi)$ is called an {\em isomorphism}.
When there exists an isomorphism between two stationary random processes,
those processes are said to be {\em isomorphic}. 
\end{definition} 

In order to determine if given two random processes are isomorphic or not, 
one of the most basic criterion is the ergodicity.
\begin{definition}
A random process $\bm{X}=\{X_n\}_{n\in\mathbb{Z}}$ determined by $({\cal X}^{\mathbb{Z}}, \mathscr{B}_{\cal X}, \mu,S)$ is called {\em ergodic} if,
for every $A \in \mathscr{B}_{\cal X}$ with $\mu(A \triangle S^{-1}A) = 0$, it holds that $\mu(A) =0$ or $\mu(A)=1$, where $\triangle$ is the symmetric difference
of sets.
\end{definition}
From the definition, we can readily verify that
ergodicity is an invariant under homomorphism (eg.~see \cite[Example I.2.12]{Shields-book}).\footnote{There may exist a homomorphism from
a non-ergodic process to an ergodic process.} 
\begin{proposition} \label{proposition:invariance-ergodic}
For two stationary random processes $\bm{X}=\{X_n\}_{n\in\mathbb{Z}}$ and $\bm{Y}=\{ Y_n\}_{n\in\mathbb{Z}}$,
suppose that there exists a homomorphism from $\bm{X}$ to $\bm{Y}$.
If $\bm{X}$ is ergodic, then $\bm{Y}$ is also ergodic.
\end{proposition}

Proposition \ref{proposition:invariance-ergodic} tells us that two random processes cannot be
isomorphic if one is ergodic and the other is non-ergodic. When both processes are ergodic, 
a more quantitative invariant is needed.
\begin{definition}
For a stationary random process $\bm{X}=\{X_n\}_{n\in \mathbb{Z}}$, the entropy rate is defined by
\begin{align*}
H(\bm{X}) 
:= \lim_{n\to\infty} \frac{1}{n} H(X_1,\ldots,X_n),
\end{align*}
where 
\begin{align*}
H(X_1,\ldots,X_n) := \sum_{x^n \in {\cal X}^n} P_{X^n}(x^n) \log \frac{1}{P_{X^n}(x^n)}.
\end{align*}
\end{definition}

One of the fundamental results in ergodic theory is the following.

\begin{proposition}[Homomorphic monotonicity of entropy \cite{Kolmogorov:58,Sinai:59}] \label{proposition:Kolmogorov-Sinai}
For two stationary random processes $\bm{X}=\{X_n\}_{n\in\mathbb{Z}}$ and $\bm{Y}=\{ Y_n\}_{n\in\mathbb{Z}}$,
if a homomorphism from $\bm{X}$ to $\bm{Y}$ exists, then it holds that
\begin{align*}
H(\bm{X}) \ge H(\bm{Y}).
\end{align*}
\end{proposition}

\begin{corollary}[Isomorphic invariance of entropy]
If two stationary random processes $\bm{X}=\{X_n\}_{n\in\mathbb{Z}}$ and $\bm{Y}=\{ Y_n\}_{n\in\mathbb{Z}}$ are isomorphic, then it holds that
\begin{align*}
H(\bm{X}) = H(\bm{Y}).
\end{align*}
\end{corollary}

The entropy has been the most widely used invariant to determine if
two random processes are isomorphic or not. In fact, when two random processes are independently identically distributed (i.i.d.) processes,
i.e., Bernoulli shifts, then Ornstein proved that the entropy is the complete invariant, i.e.,
the two processes are isomorphic if and only if their entropies are the same \cite{Ornstein:70}. 
  
%%%%%%%%%%%%%%%%%%%%%%%%%%%%%%%%%%%%%%%%%%%%%%%%%%%%%%%%%%%%%%%%%%%%%%%%%%%%
%%%%%%%%%%%%%%%%%%%%%%%%%%%%%%%%%%%%%%%%%%%%%%%%%%%%%%%%%%%%%%%%%%%%%%%%%%%%  
\section{Invariance of Information Spectrum} \label{section:information-spectrum}

Let us introduce the information spectrum of a random process \cite{han:book}.

\begin{definition} \label{definition:information-spectrum}
For a stationary random process $\bm{X}=\{X_n \}_{n\in\mathbb{Z}}$, the information spectrum 
is the cumulative distribution function of the normalized self-information defined by
\begin{align*}
F_{\bm{X}}(\tau) := \lim_{\gamma \downarrow 0} \limsup_{n\to\infty} \Pr\bigg( \frac{1}{n} \log \frac{1}{P_{X^n}(X^n)} \le \tau + \gamma \bigg)
\end{align*}
for $\tau \in \mathbb{R}^+ := \{ a \in \mathbb{R}: a \ge 0 \}$.
\end{definition}

By the definition, $F_{\bm{X}}(\tau)$ is right-continuous. Since
\begin{align*}
 \Pr\bigg( \frac{1}{n} \log \frac{1}{P_{X^n}(X^n)} > \log|{\cal X}|+\gamma \bigg) &=  \sum_{x^n \in {\cal X}^n} P_{X^n}(x^n) \bol{1}\bigg[ P_{X^n}(x^n) < \frac{2^{-n\gamma}}{|{\cal X}|^n}\bigg] \\
&\le 2^{-n\gamma}
\end{align*}
for any $\gamma>0$, it follows that $F_{\bm{X}}(\tau) = 1$ for $\tau \ge \log |{\cal X}|$.

When a random process $\bm{X}$ is ergodic, the asymptotic equipartition property guarantees 
\begin{align*}
\lim_{n \to \infty} \Pr\bigg( \bigg| \frac{1}{n} \log \frac{1}{P_{X^n}(X^n)} - H(\bm{X}) \bigg| \le \gamma \bigg) = 1
\end{align*}
for any $\gamma>0$. Thus, the information spectrum of the ergodic process is given as
\begin{align*}
F_{\bm{X}}(\tau) = \bol{1}\big[ H(\bm{X}) \le \tau \big],
\end{align*}
where $\bol{1}[\cdot]$ is the indicator function.

When a random process is not ergodic, the information spectrum can be computed
based on the entropy spectrum of the ergodic decomposition of the process as follows. The proof will be given in Section \ref{section:proof-explicit-information-spectrum}.
\begin{theorem} \label{theorem:explicit-information-spectrum}
When the ergodic decomposition of a stationary process $\bm{X}=\{X_n \}_{n\in\mathbb{Z}}$ is given as
\begin{align*}
P_{X^n}(x^n) = \int_\Theta P_{X_\theta^n}(x^n) dw(\theta)
\end{align*}
for a family of ergodic processes $\{ \bm{X}_\theta \}_{\theta \in \Theta}$ with measure $w$ on $\Theta$,
the information spectrum of the process is given as
\begin{align} \label{eq:explicit-formula}
F_{\bm{X}}(\tau) = w(\{ \theta : H(\bm{X}_\theta) \le \tau \} ).
\end{align}
\end{theorem}

Let $\overline{H}(\bm{X})$ and $\underline{H}(\bm{X})$ be defined as
\begin{align*}
\overline{H}(\bm{X}) &:= \inf\bigg\{ \tau : \lim_{n\to\infty} \Pr\bigg( \frac{1}{n} \log \frac{1}{P_{X^n}(X^n)} > \tau \bigg) = 0  \bigg\}, \\
\underline{H}(\bm{X}) &:= \sup\bigg\{ \tau : \lim_{n\to\infty} \Pr\bigg( \frac{1}{n} \log \frac{1}{P_{X^n}(X^n)} < \tau \bigg) = 0  \bigg\},
\end{align*}
which are called the spectral sup-entropy and spectral inf-entropy \cite{han:book}.\footnote{$\overline{H}(\bm{X})$ and $\underline{H}(\bm{X})$
are given by the essential supremum and the essential infimum of the entropy spectrum under the ergodic decomposition. 
The essential supremum of the entropy spectrum was used in \cite{Winkelbauer:70} to characterize the limit of 
the source coding for non-ergodic processes. } 
Then, we have $F_{\bm{X}}(\tau) = 1$ for $\tau \ge \overline{H}(\bm{X})$
and $F_{\bm{X}}(\tau) = 0$ for $\tau < \underline{H}(\bm{X})$.
When a process is decomposed into a finite number of ergodic components, 
a behavior of the information spectrum is described in Fig.~\ref{Fig:discrete}.

%%%
\begin{figure}[t]
\centering{
\includegraphics[width=0.45\textwidth]{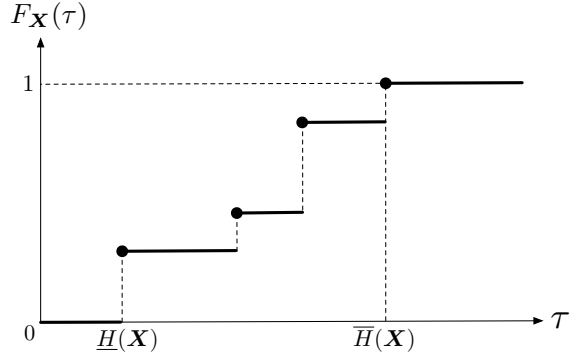}
\caption{A behavior of the information spectrum when a process is decomposed into a finite number of ergodic components.}
\label{Fig:discrete}
}
\end{figure}
%%%

\begin{remark}
If the information spectrum $F_\bm{X}(\tau)$ is defined without the slack parameter $\gamma\downarrow 0$ in Definition \ref{definition:information-spectrum}, it may not
be right-continuous in general. For instance, when $\bm{X}$ is an i.i.d. process, the law of large number and the central limit theorem imply
\begin{align*}
\lim_{n\to\infty} \Pr\bigg( \frac{1}{n} \log \frac{1}{P_{X^n}(X^n)} \le \tau  \bigg) 
= \left\{
\begin{array}{ll}
0 & \tau < H(\bm{X}) \\
1/2 & \tau = H(\bm{X}) \\
1 & \tau > H(\bm{X})
\end{array}
\right..
\end{align*}
\end{remark}

As an information spectrum counterpart of Proposition \ref{proposition:Kolmogorov-Sinai}, we have the following theorem, 
which will be proved in Section \ref{section:proof-spectrum-homomorphism}.

\begin{theorem}[Homomorphic monotonicity of information spectrum] \label{theorem:spectrum-homomorphism}
For two stationary random processes $\bm{X}=\{X_n\}_{n\in\mathbb{Z}}$ and $\bm{Y}=\{ Y_n\}_{n\in\mathbb{Z}}$,
if a homomorphism from $\bm{X}$ to $\bm{Y}$ exists, then it holds that
\begin{align*}
F_{\bm{X}}(\tau) \le F_{\bm{Y}}(\tau)
\end{align*}
for every $\tau \in \mathbb{R}^+$.
\end{theorem}

Theorem \ref{theorem:spectrum-homomorphism} says that a necessary condition for the existence of
a homomorphism is that the spectrum of $\bm{X}$ ``dominates" the spectrum of $\bm{Y}$ (cf.~Fig.~\ref{Fig:dominance}).
Simpler necessary conditions are
\begin{align*}
\overline{H}(\bm{Y}) \le \overline{H}(\bm{X})
\end{align*}
and
\begin{align*}
\underline{H}(\bm{Y}) \le \underline{H}(\bm{X}).
\end{align*}

%%%
\begin{figure}[t]
\centering{
\includegraphics[width=0.45\textwidth]{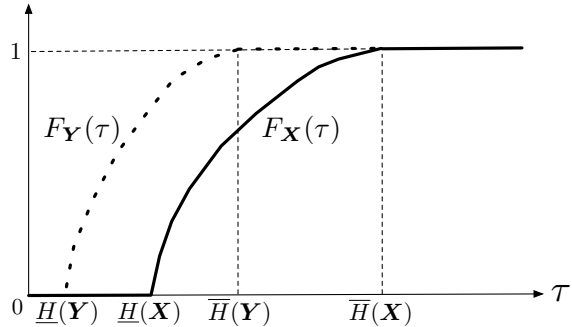}
\caption{Behaviors of the information spectrum of $\bm{X}$ (solid curve) and the information spectrum of $\bm{Y}$ (dashed curve) when there exists a homomorphism from $\bm{X}$ to $\bm{Y}$.}
\label{Fig:dominance}
}
\end{figure}
%%%

\begin{corollary}[Isomorphic invariance of information spectrum] \label{corollary-invariance-information-spectrum}
If two stationary random processes $\bm{X}=\{X_n\}_{n\in\mathbb{Z}}$ and $\bm{Y}=\{ Y_n\}_{n\in\mathbb{Z}}$ are isomorphic, then it holds that
\begin{align*}
F_{\bm{X}}(\tau) = F_{\bm{Y}}(\tau)
\end{align*}
for every $\tau \in \mathbb{R}^+$.
\end{corollary}

By combining Corollary \ref{corollary-invariance-information-spectrum} with Theorem \ref{theorem:explicit-information-spectrum}, we can recover the following result by \v{S}ujan.
\begin{corollary}[\cite{Sujan:82b}]
Suppose that a stationary random processes $\bm{X}=\{X_n\}_{n\in\mathbb{Z}}$ is decomposed as
\begin{align*}
P_{X^n}(x^n) = \int_\Theta P_{X_\theta^n}(x^n) dw(\theta)
\end{align*}
for a family of ergodic processes $\{ \bm{X}_\theta \}_{\theta \in \Theta}$ with measure $w$ on $\Theta$, and
a stationary random processes $\bm{Y}=\{ Y_n\}_{n\in\mathbb{Z}}$ is decomposed as
\begin{align*}
P_{Y^n}(y^n) = \int_\Xi P_{Y_\xi^n}(y^n) dv(\xi)
\end{align*}
for a family of ergodic processes $\{ \bm{Y}_\xi \}_{\xi \in \Xi}$ with measure $v$ on $\Xi$.
If the two stationary random processes $\bm{X}$ and $\bm{Y}$ are isomorphic, then it holds that
\begin{align*}
w(\{ \theta : H(\bm{X}_\theta) \le \tau \} ) = v(\{ \xi : H(\bm{Y}_\xi) \le \tau \})
\end{align*}
for every $\tau \in \mathbb{R}_+$.
\end{corollary}

Theorem \ref{theorem:spectrum-homomorphism} is a generalization
of Proposition \ref{proposition:Kolmogorov-Sinai} in the sense that the former implies the latter. 
In fact, by noting the ergodic decomposition of the entropy rate \cite[Theorem 3.3]{gray:book}, the identity of the expectation \cite[Eq.~(21.9)]{billingsley-book}, and Theorem \ref{theorem:explicit-information-spectrum}, we can write
\begin{align*}
H(\bm{X}) &= \int_\Theta H(\bm{X}_\theta) dw(\theta) \\
&= \int_0^\infty w(\{ \theta : H(\bm{X}_\theta) > \tau \}) d\tau \\
&= \int_0^\infty (1-F_\bm{X}(\tau)) d\tau;
\end{align*}
and similarly for $H(\bm{Y})$. 
Thus, by Theorem \ref{theorem:spectrum-homomorphism}, we have
\begin{align*}
H(\bm{X}) =  \int_0^\infty (1-F_\bm{X}(\tau)) d\tau \ge \int_0^\infty (1-F_\bm{Y}(\tau) ) d\tau = H(\bm{Y}).
\end{align*}

%%%%%%%%%%%%%%%%%%%%%%%%%%%%%%%%%%%%%%%%%%%%%%%%%%%%%%%%%%%%%%%%%%%%%%%%%%%%%
%%%%%%%%%%%%%%%%%%%%%%%%%%%%%%%%%%%%%%%%%%%%%%%%%%%%%%%%%%%%%%%%%%%%%%%%%%%%%
\section{Proof of Theorem \ref{theorem:explicit-information-spectrum}} \label{section:proof-explicit-information-spectrum}
  
When all components are i.i.d., \eqref{eq:explicit-formula} was proved in \cite[Lemma 1.4.4]{han:book}.
Exactly the same proof applies to one direction: the left hand side is  less than or equal to the right hand side in \eqref{eq:explicit-formula}.  
The opposite direction of the proof also proceeds along the line of \cite[Lemma 1.4.4]{han:book}, but it requires more complicated arguments to handle ergodic components. 
Let us first prove the former direction. 
To that end, we first note that
\begin{align*}
\Pr\bigg( \frac{1}{n} \log \frac{1}{P_{X^n}(X^n)} \le \tau + \gamma \bigg) 
= \int \Pr\bigg( \frac{1}{n} \log \frac{1}{P_{X^n}(X_\theta^n)} \le \tau + \gamma \bigg)  dw(\theta).
\end{align*}
For each $X_\theta^n \sim P_{X_\theta^n}$, by the union bound, we have
\begin{align*}
\Pr\bigg( \frac{1}{n} \log \frac{1}{P_{X^n}(X_\theta^n)} \le \tau + \gamma \bigg) 
&= \Pr\bigg( \frac{1}{n} \log \frac{1}{P_{X_\theta^n}(X_\theta^n)} + \frac{1}{n} \log \frac{P_{X_\theta^n}(X_\theta^n)}{P_{X^n}(X_\theta^n)} \le \tau + \gamma \bigg)  \\
&\le \Pr\bigg( \frac{1}{n} \log \frac{1}{P_{X_\theta^n}(X_\theta^n)} \le \tau+2\gamma \mbox{ or } \frac{1}{n} \log \frac{P_{X_\theta^n}(X_\theta^n)}{P_{X^n}(X_\theta^n)} \le - \gamma \bigg)  \\
&\le \Pr\bigg( \frac{1}{n} \log \frac{1}{P_{X_\theta^n}(X_\theta^n)} \le \tau+2\gamma \bigg) + \Pr\bigg( \frac{1}{n} \log \frac{P_{X_\theta^n}(X_\theta^n)}{P_{X^n}(X_\theta^n)} \le - \gamma \bigg).
\end{align*}
Here, the second term in the last equation can be upper bounded as
\begin{align}
\Pr\bigg( \frac{1}{n} \log \frac{P_{X_\theta^n}(X_\theta^n)}{P_{X^n}(X_\theta^n)} \le -\gamma \bigg)
&= \sum_{x^n \in {\cal X}^n} P_{X_\theta^n}(x^n) \bol{1}\bigg[ \frac{1}{n} \log \frac{P_{X_\theta^n}(x^n)}{P_{X^n}(x^n)} \le -\gamma \bigg] \nonumber \\
&\le \sum_{x^n \in {\cal X}^n} P_{X^n}(x^n) 2^{-n\gamma} \bol{1}\bigg[ \frac{1}{n} \log \frac{P_{X_\theta^n}(x^n)}{P_{X^n}(x^n)} \le -\gamma \bigg] \nonumber \\ 
&\le 2^{-n\gamma}. \label{eq:change-of-measure}
\end{align}
Thus, by the Fatou lemma, we have
\begin{align*}
\limsup_{n\to\infty} \Pr\bigg( \frac{1}{n} \log \frac{1}{P_{X^n}(X^n)} \le \tau + \gamma \bigg) 
\le \int \limsup_{n\to\infty} \Pr\bigg( \frac{1}{n} \log \frac{1}{P_{X_\theta^n}(X_\theta^n)} \le \tau+2\gamma \bigg) dw(\theta).
\end{align*}
For each $\theta \in \Theta$ with $H(\bm{X}_\theta) > \tau + 3\gamma$, the AEP with respect to the ergodic process $\bm{X}_\theta$ implies 
\begin{align*}
 \limsup_{n\to\infty} \Pr\bigg( \frac{1}{n} \log \frac{1}{P_{X_\theta^n}(X_\theta^n)} \le \tau+2\gamma \bigg) 
&\le \limsup_{n\to\infty} \Pr\bigg( \frac{1}{n} \log \frac{1}{P_{X_\theta^n}(X_\theta^n)} \le H(\bm{X}_\theta) - \gamma \bigg) \\
&=0. 
\end{align*}
Thus, we have
\begin{align*}
 \limsup_{n\to\infty} \Pr\bigg( \frac{1}{n} \log \frac{1}{P_{X_\theta^n}(X_\theta^n)} \le \tau+2\gamma \bigg)  \le w(\{ \theta : H(\bm{X}_\theta) \le \tau +3\gamma \}).
\end{align*}
Since the cumulative distribution function induced by a measure is right continuous (eg.~see \cite[Sec.~14]{billingsley-book}),
by taking the limit $\gamma\downarrow 0$, we have the desired inequality. 

To prove the opposite direction, we approximate each ergodic component $\bm{X}_\theta = \{ X_{\theta,n}\}_{n\in\mathbb{Z}}$ by
the $k$th order Markov process as 
\begin{align*}
P^{(k)}_{X_{\theta,k+1}^n|X_{\theta,1}^k}(x_{k+1}^n | x_1^k) := \prod_{i=k+1}^n P_{X_{\theta,i}|X_{\theta,i-k}^{i-1}}(x_i | x_{i-k}^{i-1}).
\end{align*}
Note that $P_{X_{\theta,i}|X_{\theta,i-k}^{i-1}}(x_i | x_{i-k}^{i-1}) = P_{X_{\theta,k+1}|X_{\theta,1}^{k}}(x_i | x_{i-k}^{i-1})$
by stationarity of $\bm{X}_\theta$. For a given sequence $x_1^n \in {\cal X}^n$, let
\begin{align*}
Q^{(k)}_{x_1^n}(a_1^{k+1}) := \frac{|\{ i : k+1 \le i \le n, x_{i-k}^i = a_1^{k+1}  \}|}{n-k}
\end{align*}
be the $k$th order Markov type (overlapping $(k+1)$-block empirical distribution). 
The Markov approximation has a nice property that, when two sequences $x_1^n$ and $\tilde{x}_1^n$
have the same $k$th order Markov type $q_k=Q^{(k)}_{x_1^n}=Q^{(k)}_{\tilde{x}_1^n}$, these sequences have
the same conditional probability
\begin{align} \label{eq:same-probability}
P^{(k)}_{X_{\theta,k+1}^n|X_{\theta,1}^k}(x_{k+1}^n | x_1^k) = P^{(k)}_{X_{\theta,k+1}^n|X_{\theta,1}^k}(\tilde{x}_{k+1}^n | \tilde{x}_1^k).
\end{align}
Denote by ${\cal T}_{q_k}^n$ the set of all sequences having the $k$th order Markov type $q_k$. 

Let us consider a mixture of the Markov approximations, i.e., 
\begin{align} \label{eq:mixed-markov-approximation}
P^{(k)}_{X_{k+1}^n|X_1^k}(x_{k+1}^n |x_1^k) := \int P^{(k)}_{X_{\theta,k+1}^n|X_{\theta,1}^k}(x_{k+1}^n | x_1^k) dw(\theta).
\end{align}
Since $P^{(k)}_{X_{k+1}^n|X_1^k}$ is a mixture of $k$th order Markov processes, it is invariant over a type class, i.e.,
\begin{align} \label{eq:same-probability-mixed}
P^{(k)}_{X_{k+1}^n|X_1^k}(x_{k+1}^n |x_1^k) = P^{(k)}_{X_{k+1}^n|X_1^k}(\tilde{x}_{k+1}^n |\tilde{x}_1^k)
\end{align}
for every $x_1^n, \tilde{x}_1^n \in {\cal T}_{q_k}^n$.

Let 
\begin{align*}
c_k := \max_{x_1^k \in {\cal X}^k: \atop P_{X_1^k}(x_1^k)>0} \log \frac{1}{P_{X_1^k}(x_1^k)}.
\end{align*}
We first note that
\begin{align*}
\Pr\bigg( \frac{1}{n} \log \frac{1}{P_{X^n}(X^n)} \le \tau + 4\gamma \bigg)
&\ge \Pr\bigg( \frac{1}{n} \log \frac{1}{P_{X^n_{k+1}|X_1^k}(X^n_{k+1}|X_1^k)} \le \tau + 4\gamma - \frac{c_k}{n} \bigg) \\
&\ge \Pr\bigg( \frac{1}{n} \log \frac{1}{P_{X^n_{k+1}|X_1^k}(X^n_{k+1}|X_1^k)} \le \tau + 3\gamma  \bigg)
\end{align*}
for sufficiently large $n$. Next, by using the above introduced mixture of Markov approximation, we have
\begin{align}
\lefteqn{ \Pr\bigg( \frac{1}{n} \log \frac{1}{P_{X^n_{k+1}|X_1^k}(X^n_{k+1}|X_1^k)} \le \tau + 3\gamma  \bigg) } \nonumber \\
&= \Pr\bigg( \frac{1}{n} \log \frac{1}{P^{(k)}_{X^n_{k+1}|X_1^k}(X^n_{k+1}|X_1^k)} + \frac{1}{n} \log \frac{P^{(k)}_{X^n_{k+1}|X_1^k}(X^n_{k+1}|X_1^k)}{P_{X^n_{k+1}|X_1^k}(X^n_{k+1}|X_1^k)} \le \tau + 3\gamma  \bigg)
 \nonumber  \\
&\ge \Pr\bigg( \frac{1}{n} \log \frac{1}{P^{(k)}_{X^n_{k+1}|X_1^k}(X^n_{k+1}|X_1^k)} \le \tau + 2\gamma  \bigg) 
 - \Pr\bigg( \frac{1}{n} \log \frac{P^{(k)}_{X^n_{k+1}|X_1^k}(X^n_{k+1}|X_1^k)}{P_{X^n_{k+1}|X_1^k}(X^n_{k+1}|X_1^k)} \ge \gamma \bigg) \nonumber \\
&\ge  \Pr\bigg( \frac{1}{n} \log \frac{1}{P^{(k)}_{X^n_{k+1}|X_1^k}(X^n_{k+1}|X_1^k)} \le \tau + 2\gamma  \bigg)  - 2^{-\gamma n}, \label{eq:proof-explicit-spectrum-1}
\end{align}
where the final inequality follows in a similar manner as \eqref{eq:change-of-measure}.

For each sequence $x_1^n \in {\cal X}^n$, define 
\begin{align*}
\Phi(x_1^n) := \big\{ \theta \in \Theta : P^{(k)}_{X_{\theta,k+1}^n|X_{\theta,1}^k}(x_{k+1}^n|x_1^k) \le 2^{\sqrt{n}} P^{(k)}_{X_{k+1}^n|X_1^k}(x_{k+1}^n|x_1^k) \big\}.
\end{align*}
By applying the Markov inequality along with \eqref{eq:mixed-markov-approximation}, we have
\begin{align*}
w(\Phi(x_1^n)) \ge 1 - 2^{-\sqrt{n}}
\end{align*}
for every $x_1^n \in {\cal X}^n$. From \eqref{eq:same-probability} and \eqref{eq:same-probability-mixed}, 
$\Phi(x_1^n) = \Phi(\tilde{x}_1^n)$ for every $x_1^n,\tilde{x}_1^n \in {\cal T}_{q_k}^n$, which enables us to express $\Phi(x_1^n)$
as $\Phi(q_k)$ with $q_k = Q^{(k)}_{x_1^n}$. Let 
\begin{align*}
\Phi_n^* := \bigcap_{q_k \in {\cal Q}_k} \Phi(q_k),
\end{align*}
where ${\cal Q}_k$ is the set of all $k$th order Markov types. Since the cardinality of ${\cal Q}_k$ can be bounded as (eg.~see \cite[Theorem 1.6.14]{Shields-book})
\begin{align*}
|{\cal Q}_k| \le L_{n,k} := (n-k+1)^{|{\cal X}|^{k+1}},
\end{align*}
by the union bound, we have
\begin{align} \label{eq:the-intersection-good-theta}
w(\Phi_n^*) \ge 1 - L_{n,k} 2^{-\sqrt{n}}.
\end{align}

By removing the complement of $\Phi_n^*$, we can evaluate the first term of \eqref{eq:proof-explicit-spectrum-1} as follows:
\begin{align}
\lefteqn{ \Pr\bigg( \frac{1}{n} \log \frac{1}{P^{(k)}_{X^n_{k+1}|X_1^k}(X^n_{k+1}|X_1^k)} \le \tau + 2\gamma  \bigg) } \nonumber \\
&= \int  \Pr\bigg( \frac{1}{n} \log \frac{1}{P^{(k)}_{X^n_{k+1}|X_1^k}(X^n_{\theta,k+1}|X_{\theta,1}^k)} \le \tau + 2\gamma  \bigg) dw(\theta) \nonumber \\
&\ge \int_{\Phi_n^*}  \Pr\bigg( \frac{1}{n} \log \frac{1}{P^{(k)}_{X^n_{k+1}|X_1^k}(X^n_{\theta,k+1}|X_{\theta,1}^k)} \le \tau + 2\gamma  \bigg) dw(\theta) \nonumber \\
&\ge \int_{\Phi_n^*}  \Pr\bigg( \frac{1}{n} \log \frac{1}{P^{(k)}_{X^n_{\theta,k+1}|X_{\theta,1}^k}(X^n_{\theta,k+1}|X_{\theta,1}^k)} \le \tau + 2\gamma -\frac{1}{\sqrt{n}} \bigg) dw(\theta) \nonumber \\
&\ge \int  \Pr\bigg( \frac{1}{n} \log \frac{1}{P^{(k)}_{X^n_{\theta,k+1}|X_{\theta,1}^k}(X^n_{\theta,k+1}|X_{\theta,1}^k)} \le \tau + 2\gamma -\frac{1}{\sqrt{n}} \bigg) dw(\theta) - L_{n,k} 2^{-\sqrt{n}}, 
\label{eq:removing-bad-theta}
\end{align}
where the second last inequality follows from the definition of $\Phi_n^*$
and the last inequality follows from \eqref{eq:the-intersection-good-theta}.
By combining the above argument along with the Fatou lemma, we have
\begin{align}
\lefteqn{ \liminf_{n\to\infty} \Pr\bigg( \frac{1}{n} \log \frac{1}{P_{X^n}(X^n)} \le \tau + 4\gamma \bigg) } \nonumber \\
&\ge \int \liminf_{k\to\infty} \liminf_{n\to\infty}  \Pr\bigg( \frac{1}{n} \log \frac{1}{P^{(k)}_{X^n_{\theta,k+1}|X_{\theta,1}^k}(X^n_{\theta,k+1}|X_{\theta,1}^k)} \le \tau + 2\gamma -\frac{1}{\sqrt{n}} \bigg) dw(\theta). 
 \label{eq:proof-explicit-spectrum-2}
\end{align}

Now, note that
\begin{align*}
\mathbb{E}\bigg[ \frac{1}{n-k} \log \frac{1}{P^{(k)}_{X^n_{\theta,k+1}|X_{\theta,1}^k}(X^n_{\theta,k+1}|X_{\theta,1}^k)} \bigg]
&= \mathbb{E}\bigg[ \frac{1}{n-k} \sum_{i=k+1}^n \log \frac{1}{P_{X_{\theta,k+1}|X_{\theta,1}^k}(X_{\theta,i}|X_{\theta,i-k}^{i-1})} \bigg] \\
&= H(X_{\theta,k+1}|X_{\theta,1}^k).
\end{align*}
Here, $H(X_{\theta,k+1}|X_{\theta,1}^k)$ is nondecreasing and converges to the entropy rate $H(\bm{X}_\theta)$
(cf.~\cite[Theorem 4.2.1, Theorem 4.2.2]{cover}). Thus, when $\tau \ge H(\bm{X}_\theta)$, by taking sufficiently large $k$, we have
$\tau \ge H(X_{\theta,k}|X_{\theta,1}^k)-\gamma$. For such $k$, by the ergodic theorem, we have
\begin{align*}
\lefteqn{ \liminf_{n\to\infty} \Pr\bigg( \frac{1}{n} \log \frac{1}{P^{(k)}_{X^n_{\theta,k+1}|X_{\theta,1}^k}(X^n_{\theta,k+1}|X_{\theta,1}^k)} \le \tau + 2\gamma -\frac{1}{\sqrt{n}} \bigg) } \\
&\ge \liminf_{n\to\infty} \Pr\bigg( \frac{1}{n} \log \frac{1}{P^{(k)}_{X^n_{\theta,k+1}|X_{\theta,1}^k}(X^n_{\theta,k+1}|X_{\theta,1}^k)} \le H(X_{\theta,k}|X_{\theta,1}^k) + \gamma -\frac{1}{\sqrt{n}} \bigg)
\\&=1.
\end{align*}
By combining it with \eqref{eq:proof-explicit-spectrum-2}, we have
\begin{align*}
 \liminf_{n\to\infty} \Pr\bigg( \frac{1}{n} \log \frac{1}{P_{X^n}(X^n)} \le \tau + 4\gamma \bigg) \ge w(\{ \theta : H(\bm{X}_\theta) \le \tau \}).
\end{align*}
Since this inequality holds for any $\gamma>0$, we have the desired inequality. \qed  
  
%%%%%%%%%%%%%%%%%%%%%%%%%%%%%%%%%%%%%%%%%%%%%%%%%%%%%%%%%%%%%%%%%%%%%%%%%%%%
%%%%%%%%%%%%%%%%%%%%%%%%%%%%%%%%%%%%%%%%%%%%%%%%%%%%%%%%%%%%%%%%%%%%%%%%%%%%
\section{Proof of Theorem \ref{theorem:spectrum-homomorphism}} \label{section:proof-spectrum-homomorphism}
  
Let $\bm{X}=\{X_n\}_{n\in\mathbb{Z}}$ and $\bm{Y}=\{Y_n\}_{n\in\mathbb{Z}}$ be stationary random processes
determined by $({\cal X}^{\mathbb{Z}}, \mathscr{B}_{\cal X}, \mu,S)$ and $({\cal Y}^{\mathbb{Z}}, \mathscr{B}_{\cal Y},\nu,S)$.
For a given homomorphism $\phi : {\cal X}^\mathbb{Z} \to {\cal Y}^\mathbb{Z}$ and integers $m,n$, we can construct a coupling
\begin{align} \label{eq:coupling}
P_{X_{-m}^m Y_{-n}^n}(x_{-m}^m,y_{-n}^n) = \mu\big( [x_{-m}^m] \cap \phi^{-1}([y_{-n}^n]) \big)
\end{align}  
induced by the process $\bm{X}$ and the homomorphism $\phi$.
Since $\nu = \phi_* \mu$, the marginal of the joint distribution $P_{X_{-m}^m Y_{-n}^n}$ coincides with
the distribution induced by the process $\bm{Y}$. In the following argument, by a slight abuse of notation,
we interpret that the random variable $(X_{-m}^m,Y_{-n}^n)$ are distributed according to the joint distribution 
given by \eqref{eq:coupling}.  
  
First, we approximate an arbitrary homomorphism by using a finite function (eg.~see \cite[Theorem 3.1]{Gray:75} or \cite[Theorem 1.8.1]{Shields-book}).
\begin{lemma} \label{lemma:approximation-homomorphism}
For a given homomorphism $\phi$ from $\bm{X}$ to $\bm{Y}$
and arbitrary $\varepsilon > 0$, there exists an integer $\ell=\ell(\varepsilon)$ and a finite function $f:{\cal X}^{2\ell+1}\to {\cal Y}$ such that
\begin{align}
\Pr\big( Y_0 \neq f(X_{-\ell}^\ell) \big)  \le \varepsilon. \label{eq:approximation-homomorphism}
\end{align}
\end{lemma}

Let us now fix an integer $n$, and set $N=2n+1$. By stationarity, \eqref{eq:approximation-homomorphism} implies 
\begin{align} \label{eq:symbol-error-probability}
\mathbb{E}\bigg[ \frac{1}{N} d_H(Y_{-n}^n, \tilde{Y}_{-n}^n) \bigg] = \frac{1}{N} \sum_{i=-n}^n \Pr\big( Y_i \neq \tilde{Y}_i \big) \le \varepsilon,
\end{align}
where $\tilde{Y}_i = f(X_{i-\ell}^{i+\ell})$ and $d_H(\cdot,\cdot)$ is the Hamming distance.
Furthermore, by the Markov inequality, \eqref{eq:symbol-error-probability} implies 
\begin{align} \label{eq:excess-Haming-distortion}
\Pr\bigg( \frac{1}{N} d_H(Y_{-n}^n, \tilde{Y}_{-n}^n) > \beta \bigg) \le \frac{\varepsilon}{\beta}
\end{align}
for arbitrary $\beta > 0$.

The next lemma is the most key part of the proof of Theorem \ref{theorem:spectrum-homomorphism}. 
%%%
\begin{lemma} \label{lemma:finite-bound}
For a given homomorphism $\phi$ from $\bm{X}$ to $\bm{Y}$ and arbitrary $\varepsilon>0$,
let $\ell$ and $f$ be the integer and finite function specified by Lemma \ref{lemma:approximation-homomorphism}.
For an integer $n$, set $m=n+\ell$ and $N=2n+1$. Then, we have
\begin{align*}
\Pr\bigg( \frac{1}{N} \log \frac{1}{P_{X_{-m}^m}(X_{-m}^m)} \le \tau + \gamma \bigg) \le \Pr\bigg( \frac{1}{N} \log \frac{1}{P_{Y_{-n}^n}(Y_{-n}^n)} \le \tau + 2 \gamma \bigg)
 + \frac{\varepsilon}{\beta} + 2^{-N(\gamma-h(\beta)-\beta \log |{\cal Y}|)}
\end{align*}
for any $\tau \in \mathbb{R}^+$, $\gamma>0$, and $\beta>0$, where $h(\cdot)$ is the binary entropy function.
\end{lemma}
%%%
\begin{proof}
Let
\begin{align*}
{\cal S} := \bigg\{ (x_{-m}^m,y_{-n}^n) \in {\cal X}^M \times {\cal Y}^N: \frac{1}{N} \log \frac{1}{P_{X_{-m}^m}(x_{-m}^m)} \le \frac{1}{N} \log \frac{1}{P_{Y_{-n}^n}(y_{-n}^n)} - \gamma \bigg\}
\end{align*}
and
\begin{align*}
{\cal C} := \bigg\{ (x_{-m}^m,y_{-n}^n) \in {\cal X}^M \times {\cal Y}^N : d_H(y_{-n}^n, f_{-n}^n(x_{-m}^m)) \le N\beta \bigg\},
\end{align*}
where $M=2m+1$ and $f_{-n}^n(x_{-m}^m) = (f(x_{-m}^{-n+\ell}),\ldots,f(x_{-\ell}^\ell),\ldots,f(x_{n-\ell}^m))$.
For the joint distribution given by \eqref{eq:coupling}, we have
\begin{align}
P_{X_{-m}^m Y_{-n}^n}({\cal S}) &= P_{X_{-m}^m Y_{-n}^n}({\cal S} \cap {\cal C}^c) + P_{X_{-m}^m Y_{-n}^n}({\cal S} \cap {\cal C}) \nonumber \\
&\le P_{X_{-m}^m Y_{-n}^n}({\cal C}^c) + P_{X_{-m}^m Y_{-n}^n}({\cal S} \cap {\cal C}) \nonumber \\
&\le \frac{\varepsilon}{\beta} + P_{X_{-m}^m Y_{-n}^n}({\cal S} \cap {\cal C}), \label{eq:proof-key-lemma-1}
\end{align}
where ${\cal C}^c$ is the complement of ${\cal C}$ and the last inequality follows from \eqref{eq:excess-Haming-distortion}.

To evaluate the second term of \eqref{eq:proof-key-lemma-1}, note that $(x_{-m}^m,y_{-n}^n) \in {\cal S}$ implies
\begin{align} \label{eq:proof-key-lemma-2}
P_{Y_{-n}^n}(y_{-n}^n) \le 2^{-N\gamma} P_{X_{-m}^m}(x_{-m}^m);
\end{align}
also note that, for fixed $x_{-m}^m \in {\cal X}^M$, 
\begin{align} \label{eq:proof-key-lemma-3}
|\{ y_{-n}^n \in {\cal Y}^N : (x_{-m}^m,y_{-n}^n) \in {\cal C} \}| \le |{\cal Y}|^{N\beta} 2^{Nh(\beta)}
\end{align}
holds for $0 < \beta < 1/2$.
By noting these facts, we have
\begin{align}
P_{X_{-m}^m Y_{-n}^n}({\cal S} \cap {\cal C}) &= \sum_{(x_{-m}^m,y_{-n}^n) \in {\cal S}\cap {\cal C}} P_{X_{-m}^m Y_{-n}^n}(x_{-m}^m,y_{-n}^n) \nonumber \\
&\le \sum_{(x_{-m}^m,y_{-n}^n) \in {\cal S}\cap {\cal C}} P_{Y_{-n}^n}(y_{-n}^n) \nonumber \\
&\le 2^{-N\gamma} \sum_{(x_{-m}^m,y_{-n}^n) \in {\cal S}\cap {\cal C}} P_{X_{-m}^m}(x_{-m}^m) \nonumber \\
&\le 2^{-N\gamma} \sum_{x_{-m}^m \in {\cal X}^M} P_{X_{-m}^m}(x_{-m}^m) \sum_{y_{-n}^n \in {\cal Y}^N} \bol{1}[ (x_{-m}^m,y_{-n}^n) \in {\cal C}] \nonumber \\
&\le 2^{-N(\gamma-h(\beta)-\beta\log|{\cal Y}|)},  \label{eq:proof-key-lemma-4}
\end{align}
where the second inequality follows from \eqref{eq:proof-key-lemma-2}
and the last inequality follows from \eqref{eq:proof-key-lemma-3}.

Finally, note that 
\begin{align}
P_{X_{-m}^m Y_{-n}^n}({\cal S}) &= \Pr\bigg( \frac{1}{N} \log \frac{1}{P_{X_{-m}^m}(X_{-m}^m)} \le \frac{1}{N} \log \frac{1}{P_{Y_{-n}^n}(Y_{-n}^n)} - \gamma \bigg) \nonumber \\
&\ge \Pr\bigg( \frac{1}{N} \log \frac{1}{P_{X_{-m}^m}(X_{-m}^m)} \le \tau + \gamma,~\frac{1}{N} \log \frac{1}{P_{Y_{-n}^n}(Y_{-n}^n)} > \tau +2 \gamma \bigg) \nonumber \\
&\ge \Pr\bigg( \frac{1}{N} \log \frac{1}{P_{X_{-m}^m}(X_{-m}^m)} \le \tau + \gamma \bigg)
- \Pr\bigg( \frac{1}{N} \log \frac{1}{P_{Y_{-n}^n}(Y_{-n}^n)} \le \tau +2 \gamma \bigg)  \label{eq:proof-key-lemma-5}
\end{align}
for any $\tau \in \mathbb{R}^+$.
By combining \eqref{eq:proof-key-lemma-1}, \eqref{eq:proof-key-lemma-4}, and \eqref{eq:proof-key-lemma-5}, we have the claim of the lemma.
\end{proof}

Let $M := 2m+1 = N + 2\ell$. In order to replace $N$ with $M$ in the left hand side of the bound in Lemma \ref{lemma:finite-bound},
take $n$ sufficiently large so that $N(\tau+\gamma)/(N+2\ell) \ge \tau+\gamma/2$.
Then, the bound in Lemma \ref{lemma:finite-bound} implies 
\begin{align*}
\Pr\bigg( \frac{1}{M} \log \frac{1}{P_{X_{-m}^m}(X_{-m}^m)} \le \tau + \gamma/2 \bigg) \le \Pr\bigg( \frac{1}{N} \log \frac{1}{P_{Y_{-n}^n}(Y_{-n}^n)} \le \tau + 2 \gamma \bigg)
 + \frac{\varepsilon}{\beta} +  2^{-N(\gamma-h(\beta)-\beta \log |{\cal Y}|)}.
\end{align*}
Now, take $\beta=\beta_\gamma$ sufficiently small compared to $\gamma$ so that the exponent of the last term becomes positive; 
then by taking the limit of $n$ and by noting the stationarity of $\bm{X}$ and $\bm{Y}$, we have
\begin{align*}
\limsup_{n\to\infty} \Pr\bigg( \frac{1}{n} \log \frac{1}{P_{X^n}(X^n)} \le \tau + \gamma/2 \bigg) \le \limsup_{n\to\infty} \Pr\bigg( \frac{1}{n} \log \frac{1}{P_{Y^n}(Y^n)} \le \tau + 2\gamma \bigg) + \frac{\varepsilon}{\beta_\gamma}.
\end{align*}
Since this inequality holds for arbitrary $\varepsilon>0$,\footnote{Note that the other terms do not depend on $\ell$ anymore.} by taking the limit $\varepsilon \downarrow 0$, we have
\begin{align*}
\limsup_{n\to\infty} \Pr\bigg( \frac{1}{n} \log \frac{1}{P_{X^n}(X^n)} \le \tau + \gamma/2 \bigg) \le \limsup_{n\to\infty} \Pr\bigg( \frac{1}{n} \log \frac{1}{P_{Y^n}(Y^n)} \le \tau + 2\gamma \bigg).
\end{align*}
Finally, by taking the limit $\gamma\downarrow 0$, we have the desired result. \qed  
  
%%%%%%%%%%%%%%%%%%%%%%%%%%%%%%%%%%%%%%%%%%%%%%%%%%%%%%%%%%%%%%%%%%%%%%%%%%%%  
%%%%%%%%%%%%%%%%%%%%%%%%%%%%%%%%%%%%%%%%%%%%%%%%%%%%%%%%%%%%%%%%%%%%%%%%%%%%
%\section{Discussion} \label{section:discussion}

\section{General Dynamical System} \label{sec:general-dynamical-system}

For a general dynamical system $(\Omega, \mathscr{B},\lambda,T)$, the entropy is defined via 
homomorphism from the dynamical system to a finite alphabet random process (eg.~see \cite{gray:book}).\footnote{It is more common
to define the entropy of a dynamical system via a partition, but they are equivalent.} 
Let $\phi$ be a homomorphism from $(\Omega, \mathscr{B},\lambda,T)$ to a finite alphabet random process $\bm{X}$ determined by
$({\cal X}^{\mathbb{Z}}, \mathscr{B}_{\cal X}, \mu,S)$, where $\mu = \phi_*\lambda$.
Using such a homomorphism, the entropy of the dynamical system is defined by
\begin{align*}
H(\lambda,T) := \sup_\phi H(\bm{X}),
\end{align*}
where the supremum is taken over all homomorphisms from the dynamical system $(\Omega, \mathscr{B},\lambda,T)$ to finite alphabet random processes.
Even though computing the entropy of a dynamical system is difficult in general, thanks to Proposition \ref{proposition:Kolmogorov-Sinai},
we can compute the entropy once we find an isomorphism from the dynamical system to a finite alphabet random process. 

In a similar spirit, we define the information spectrum of a dynamical system as follows:
\begin{align*}
F_{\lambda,T}(\tau) := \inf_\phi F_{\bm{X}}(\tau),
\end{align*}
where the infimum is takin over all homomorphisms from the dynamical system $(\Omega, \mathscr{B},\lambda,T)$ to finite alphabet random 
processes.\footnote{It may be possible that $F_{\lambda,T}(\tau) = 0$ for every $\tau \in \mathbb{R}^+$. For this reason, we also define $F_{\lambda,T}(\infty):=1$.}
Again, it is difficult to compute the information spectrum of a dynamical system in general.
However, thanks to Theorem \ref{theorem:spectrum-homomorphism}, we can compute the information spectrum
once we find an isomorphism from the dynamical system to a finite alphabet random process. 

%%%%%%%%%%%
\section{Sufficient Condition} \label{section:sufficient-condition}

A stationary random process $\bm{X}$ is termed a {\em B-processes} if it is a stationary coding of an i.i.d. process,
i.e., there exists a homomorphism from an i.i.d. process to $\bm{X}$  (eg.~see \cite{Ornstein:book, Shields-book}). 
In a series of papers \cite{Ornstein:70, Ornstein:71} (see also \cite{Ornstein:book}), Ornstein showed that
any ``finitely determined" process is isomorphic to an i.i.d. process having the same entropy;\footnote{Finitely determined is
a property such that approximation in the sense of the total variational distance and the entropy 
implies approximation in the sense of the $\bar{d}$-distance.} 
and that any B-process is finitely determined (see also \cite[Chapter IV]{Shields-book} for other characterizations of the B-process). 
Consequently, the theory by Ornstein says that the class of B-processes can be classified by the entropy, i.e., two B-processes having the
same entropy are isomorphic to each other. 
Then, it is tempting to extend this classification theory to mixtures of B-processes by using the information spectrum.
However, there are some pathological cases, which will be discussed later, and all mixtures of B-processes cannot necessarily be classified only by the information spectrum.
In the following, let us confine ourselves to the following class of processes. 

\begin{definition}[Countable regular mixture of B-Process]
A stationary random process $\bm{X}$ determined by $({\cal X}^{\mathbb{Z}}, \mathscr{B}_{\cal X}, \mu,S)$ is termed a {\em countable regular mixture of B-processes} if
the ergodic decomposition is given by
\begin{align*}
\mu(B) = \sum_{\theta \in \mathbb{N}} w(\theta) \mu_\theta(B) ,~B \in {\cal B}_{\cal X}
\end{align*}
with a family of B-process $\{ \bm{X}_\theta \}_{\theta \in \mathbb{N}}$ determined by $({\cal X}^{\mathbb{Z}}, \mathscr{B}_{\cal X}, \mu_\theta,S)$
(with measure $w$ on the set of integers $\mathbb{N}$). Here, it is assumed that $H(\bm{X}_\theta) \neq H(\bm{X}_{\theta^\prime})$ for every $\theta \neq \theta^\prime$
(regularity condition) is satisfied.
\end{definition}

This class of processes can be classified by the information spectrum as follows:
\begin{proposition} \label{proposition:sufficiency}
Suppose that $\bm{X}$ and $\bm{Y}$ are countable regular mixtures of B-processes, and $F_{\bm{X}}(\tau) = F_{\bm{Y}}(\tau)$
for every $\tau \in \mathbb{R}_+$. Then, there exists an isomorphism between $\bm{X}$ and $\bm{Y}$.
\end{proposition}
%%%
\begin{proof}
Let $({\cal X}^{\mathbb{Z}}, \mathscr{B}_{\cal X}, \mu_\theta,S)$ with $w$ and $({\cal Y}^{\mathbb{Z}}, \mathscr{B}_{\cal Y}, \nu_\xi,S)$
with $v$ be the ergodic decompositions of $\bm{X}$ and $\bm{Y}$, respectively. Since
Theorem \ref{theorem:explicit-information-spectrum} and $F_{\bm{X}}(\tau) = F_{\bm{Y}}(\tau)$ imply 
\begin{align*}
w(\{ \theta : H(\bm{X}_\theta) \le \tau \}) = v(\{ \xi : H(\bm{Y}_\xi) \le \tau \}) 
\end{align*} 
for every $\tau \in \mathbb{R}_+$ and $\bm{X}$ and $\bm{Y}$ are regular mixtures, there exists one-to-one mapping $\kappa$ such that
$H(\bm{X}_\theta) = H(\bm{Y}_{\kappa(\theta)})$ and $w(\theta)=v(\kappa(\theta))$ for $\theta \in \mathbb{N}$. To avoid cumbersome notation,
without loss of generality, we assume that $\kappa$ is the identity in the following.

Since each component $\bm{X}_\theta$ and $\bm{Y}_\theta$ are B-processes having the same entropy, Ornstein's isomorphism theorem (cf.~\cite{Ornstein:book}) implies the existence 
of an isomorphism $(\phi_\theta,\psi_\theta)$ such that $\psi_\theta(\phi_\theta(\bm{x}))=\bm{x}$
for every $\bm{x} \in C_\theta$ for some $C_\theta$ satisfying $\mu_\theta(C_\theta)=1$ and $\phi_\theta(\psi_\theta(\bm{y}))=\bm{y}$ for
every $\bm{y} \in D_{\theta}$ for some $D_\theta$ satisfying $\nu_{\theta}(D_\theta)=1$. We construct an isomorphism between $\bm{X}$ and $\bm{Y}$
by pasting these isomorphisms together. 

By a well known fact on the ergodic decomposition (eg.~see \cite{Shields-book}), there exists a disjoint family $\{ A_\theta\}$ such that $A_\theta$ is shift invariant, i.e., $SA_\theta =A_\theta$,
$\mu_\theta(A_\theta)=1$, and $\mu_{\theta^\prime}(A_\theta)=0$ for $\theta^\prime \neq \theta$; similarly, there exists a disjoint family
$\{B_\theta \}$ such that $B_\theta$ is shift invariant, $\nu_\theta(B_\theta)=1$, and $\nu_{\theta^\prime}(B_\theta)=0$ for $\theta^\prime\neq \theta$.

Let $\phi_{\theta,0}(\bm{x})= (\phi_\theta(\bm{x}))_0$. Let
\begin{align*}
\phi_0(\bm{x}) = \left\{
\begin{array}{ll}
y & \mbox{if } \bm{x} \in \cup_{\theta \in \mathbb{N}} ( A_\theta \cap \phi_{\theta,0}^{-1}(y)) \mbox{ for } y \in {\cal Y} \\
b & \mbox{if } \bm{x} \notin \cup_{\theta \in \mathbb{N}} A_\theta
\end{array}
\right.,
\end{align*}
where $b \in {\cal Y}$ is an arbitrary constant.
Since a countable union of measurable sets is measurable, $\phi_0$ defined above is measurable.
Let $(\phi(\bm{x}))_i = \phi_0(S^i \bm{x})$. 
Since $A_\theta$ is shift invariant
and $S^i \bm{x} \in \phi_{\theta,0}^{-1}(y)$ implies $(\phi_\theta(S^i \bm{x}))_0 = (\phi_\theta(\bm{x}))_i=y$ for each $i \in \mathbb{Z}$, 
we have $\phi(\bm{x})=\phi_\theta(\bm{x})$ for $\bm{x} \in A_\theta$.
Similarly, we can construct $\psi$ from $\{ \psi_\theta\}$ such that $\psi(\bm{y})=\psi_\theta(\bm{y})$ for $\bm{y} \in B_\theta$.
Then, for $\bm{x} \in A_\theta \cap \phi_\theta^{-1}(B_\theta) \cap C_\theta$, we have
\begin{align*}
\psi(\phi(\bm{x})) = \psi_\theta(\phi_\theta(\bm{x}))= \bm{x}.
\end{align*}
Furthermore, we have
\begin{align*}
\mu\bigg( \bigcup_{\theta \in \mathbb{N}} ( A_\theta \cap \phi_\theta^{-1}(B_\theta) \cap C_\theta) \bigg) 
&= \sum_{\theta \in \mathbb{N}} \mu(A_\theta \cap \phi_\theta^{-1}(B_\theta) \cap C_\theta) \\
&= \sum_{\theta \in \mathbb{N}} w(\theta) \mu_\theta( \phi_\theta^{-1}(B_\theta) \cap C_\theta) \\
&=1,
\end{align*}
where the last equality follows from $\mu_\theta(C_\theta)=1$ and $\mu_\theta(\phi_\theta^{-1}(B_\theta))=\nu_\theta(B_\theta)=1$.
Thus, $\psi(\phi(\bm{x}))=\bm{x}$ for $\mu$ almost every $\bm{x}$. Similarly, $\phi(\psi(\bm{y}))=\bm{y}$ for $\nu$ almost every $\bm{y}$.
Finally, for $B \in \mathscr{B}_{\cal Y}$, we have
\begin{align*}
\mu(\phi^{-1}(B)) &=  \sum_{\theta \in \mathbb{N}} w(\theta) \mu_\theta(\phi^{-1}(B)) \\
&= \sum_{\theta \in \mathbb{N}} w(\theta) \mu_\theta(A_\theta \cap \phi^{-1}(B)) \\
&= \sum_{\theta \in \mathbb{N}} w(\theta) \mu_\theta(A_\theta \cap \phi^{-1}_\theta(B)) \\
&= \sum_{\theta \in \mathbb{N}} w(\theta) \mu_\theta(\phi^{-1}_\theta(B)) \\
&= \sum_{\theta \in \mathbb{N}} v(\theta) \nu_\theta(B) \\
&= \nu(B),
\end{align*}
i.e., $\nu=\phi_*\mu$. Similarly, we have $\mu=\psi_*\nu$. Thus, $(\phi,\psi)$ is the desired isomorphism between $\bm{X}$ and $\bm{Y}$.
\end{proof}

It is claimed in \cite[Theorem 2]{Sujan:82b} that Proposition \ref{proposition:sufficiency} holds with
neither the countability assumption nor the regularity assumption.\footnote{More precisely, only mixtures of i.i.d. processes are
considered in \cite[Theorem 2]{Sujan:82b}, but, by virtue of the Ornstein isomorphism theorem, B-processes and i.i.d. processes are
essentially the same.}
Even though the countability assumption in Proposition \ref{proposition:sufficiency} may be dispensed 
but then with more complicated arguments,\footnote{In order to handle
a mixture with uncountable ergodic components, we need to identify a universal isomorphism in the manner of \cite{KieRah81}.}
the regularity assumption, i.e., $H(\bm{X}_\theta) \neq H(\bm{X}_{\theta^\prime})$ for every $\theta \neq \theta^\prime$,
is crucial. For instance, let $\bm{X}_1$ and $\bm{X}_2$ be different ergodic processes having the same entropy
$H(\bm{X}_1)=H(\bm{X}_2)=a$, and $\bm{X}$ be a mixture of the two processes; let $\bm{Y}$ be another ergodic process with
$H(\bm{Y})= a$. Then, $\bm{X}$ and $\bm{Y}$ have the same information spectrum 
$F_\bm{X}(\tau)=F_\bm{Y}(\tau)=\bol{1}[a \le \tau]$. However, these processes cannot be isomorphic since
$\bm{X}$ is non-ergodic while $\bm{Y}$ is ergodic (cf.~Proposition \ref{proposition:invariance-ergodic}).
Thus, the claim in \cite[Theorem 2]{Sujan:82b} has a flaw.

%%%%%%%%%%%%%
\section{Discussion} \label{section:discussion}

In this paper, we proved that the information spectrum of random processes is invariant under isomorphisms.
Our proof is based on the information spectrum approach developed in information theory.
The proof of the invariance and the analysis of the information spectrum are conducted separately in two steps:
the ergodic decomposition nor the ergodic theorem are not used in the first step, and they are only used in the second step.
In some sense, this is a first attempt of applying the information spectrum approach to ergodic theory. 

On the other hand, known constructions of isomorphisms (or homomorphisms) heavily rely on the ergodic decomposition
and ergodicity of each component. Of course, since the ergodicity is preserved under isomorphisms, it is hopeless to construct
isomorphisms without using ergodicity at all. However, it is worthwhile to pursue a construction that separates the use of ergodicity
as much as possible. Such an approach will provide new insights into ergodic theory.

\section*{Acknowledgement}

The authors would like to thank Vincent Tan for comments.

%This work is supported
%in part by JSPS KAKENHI Grant Number 16H06091.

\bibliographystyle{./IEEEtranS}
\bibliography{../../../09-04-17-bibtex/reference}

\end{document}